\documentclass[12pt]{article}
\usepackage{graphicx}
\usepackage[top=2cm, bottom=2cm, left=2cm, right=2cm]{geometry}
\usepackage{amsmath}
\usepackage{amssymb}
\usepackage{mathrsfs}
\usepackage{amsthm}
\usepackage{color}
\usepackage{colortbl}
\usepackage{fdsymbol}
\usepackage{graphicx}
\usepackage{epstopdf}
\usepackage{authblk}
\def\be{\begin{equation}}   
\def\ee{\end{equation}}
\def\red#1{{\color{black} #1}} 
 
\def\={\;=\;}    
\def\+{\,+\,}    
\def\binom#1#2{\Bigl({#1\atop#2}\Bigr)} 

\title{Information thermodynamics of financial markets:\\ the Glosten-Milgrom model}
\author[1]{L\'eo Touzo}
\author[2]{Matteo Marsili}
\author[3,2]{Don Zagier}
\affil[1]{{\small Department of Physics, \'Ecole Normale Sup\'erieure, 24 rue Lhomond, 75005 Paris,  France}}
\affil[2]{{\small The Abdus Salam International Centre for Theoretical Physics, Strada Costiera 11, 34151 Trieste, Italy}}
\affil[3]{{\small Max Planck Institute for Mathematics, Vivatsgasse 7, 53111 Bonn, Germany }}
\date{July 2020}

\begin{document}
\maketitle

\begin{abstract}
  The Glosten-Milgrom model describes a single asset market, where
  informed traders interact with a market maker, in the presence of
  noise traders. {We derive an analogy between this financial model and 
  a Szil\'ard information engine by {\em i)} showing that the optimal work
  extraction protocol in the latter coincides with the pricing strategy of the
  market maker in the former and  {\em ii)} defining a market analogue of the physical 
  temperature from
  the analysis of the distribution of market orders. Then we show that the 
  expected gain of informed traders is bounded above by the product of this  
  market temperature with the amount of information that informed traders have, 
  in exact analogy with the corresponding formula for the maximal expected
  amount of work that can be extracted from a cycle of the information engine.}
  \red{This suggests that} recent ideas from information thermodynamics may 
  shed light on financial markets, and lead to generalised
  inequalities, in the spirit of the extended second law of thermodynamics.
\end{abstract}

Information plays a key role in the working of financial markets. News
about the performance of a company drives the activity of traders in
the market, and as a result, the corresponding stock price adjusts in
such a way as to reflect this information. The standard argument in
economics \cite{cochrane2009asset} maintains that prices adjust to the
discounted expected value of the future stream of payoffs that a share
of that stock will deliver. In this ideal picture, (equilibrium)
prices satisfy the efficient market hypothesis
\cite{fama1960efficient,samuelson2016proof}, i.e. they reflect all
available information, and no arbitrage (i.e risk-less gain) should be
possible. This has been the dominating paradigm in finance. Several
theoretical models have been proposed to show how information
aggregation occurs, under specific market
mechanisms~\cite{kyle1985continuous,glosten1985bid,gueant2016financial,brunnermeier2001asset,bouchaud2018trades},
and how markets become informationally
efficient~\cite{berg2001statistical}. Empirical studies suggest that
the efficient market
hypothesis~\cite{malkiel2003efficient,RJShiller,bouchaud2009markets}
is typically satisfied in real markets, in spite of some
anomalies~\cite{schwert2003anomalies}.

This would suggest that information theory should play a major role in
the theory of financial markets. But it does not.  This is reminiscent
of the state of affairs in statistical physics {a few} decades ago. 
Only very recent advances in stochastic
thermodynamics~\cite{van2015ensemble,jarzynski2011equalities,seifert2012stochastic}
have led to recognise the relation between the physics of
thermodynamic transformations and information
theory~\cite{parrondo2015thermodynamics}. There are good reasons to
expect that a similar connection with the theory of financial markets
also exists. Indeed, the no-arbitrage hypothesis and the second law of
thermodynamics both have the conceptual flavour of a no-free-lunch
statement: the first maintains that no risk-less gain can be extracted
from financial trading, the second {asserts} that no work can be
extracted from a cyclic transformation of a thermodynamic system at
constant temperature. The second law of thermodynamics refers to the
average work over many cycles. Yet, work can be extracted in a single
{realisation}  of the engine's cycle with some
probability~\cite{jarzynski2011equalities}. Likewise, we know that
statistical arbitrages do exist in real markets, but they're rare and
they do not persist in the long run~\cite{schwert2003anomalies,bouchaud2018trades}. Also,
work can be extracted from a system, if one can acquire some
information on its microscopic state, through a measurement. Likewise,
an informed trader can extract a positive gain from trading.
Furthermore, the no-arbitrage hypothesis in finance is formalised
within the theory of martingales, which has recently found several
applications in stochastic thermodynamics (see
e.g. \cite{neri2019integral,chetrite2019martingale}). This suggests a
common conceptual basis for finance and stochastic thermodynamics.

In order to explore this connection, this paper focuses on the
Glosten-Milgrom (GM) model~\cite{glosten1985bid}, which provides a
simple setting for a detailed analysis. This describes a single asset
market where a population or informed and non-informed traders
interact with a market maker, who sets the prices.

We derive a general analogy between the GM model and the Szil\'ard information
engine~\cite{szilard-german,parrondo2001szilard,kim2011information},
which is a prototype model in information thermodynamics. Within this
analogy, we identify {a ``market temperature"} $T$ (Eq.~\eqref{T}), which is
the {analogue} in the GM model of the temperature of the heat bath in 
the Szil\'ard engine. We also show that
the optimal pricing policy of the market maker in the GM model
corresponds to the optimal work extraction protocol in the Szil\'ard
engine. The sequence of trades in the financial system corresponds to a 
measurement process in the physics analogy. 
The generalised second law of 
thermodynamics~\cite{sagawa2012thermodynamics} then states that, in the physical system, 
the maximal expected value of the work $W$ that can be extracted
from a cycle of the engine equals $T$ times the amount of information $H[Y]$ 
provided by the measurement, i.e., 
\begin{equation}   \label{ineqW}
  {\mathbb{E}[W] \,\le\, TH[Y]}\,.
\end{equation}
In the financial system $H[Y]$ quantifies the amount of information
that informed traders have on the value of the asset. In exact analogy with Eq.~\eqref{ineqW},
we prove that the expected value of the gain $G$ of informed
traders in the GM model is bounded above by the product of \red{the market temperature} 
$T$ and $H[Y]$, i.e.,
\begin{equation}
  \label{Gineq}
  {\mathbb{E}[G] \,\le\, TH[Y]}\,.
\end{equation}
This is the main result of this paper. We also show that the
inequality (\ref{Gineq}) becomes an asymptotic equality in the limit  when 
noise traders dominate the market, which corresponds to the limit $T\to\infty$.

The similarity between inequalities~(\ref{ineqW}) and~(\ref{Gineq}) puts the analogy
between thermodynamics and finance discussed above on firmer grounds, suggesting
a common conceptual basis for both.

The inequality~(\ref{Gineq}) is reminiscent of the classical
result~\cite{cover1999elements} relating the growth rate of the
capital to the information that investor has on the odds, in a scheme
of lotteries. The connection between this result and stochastic
thermodynamics has been explored by several
authors~\cite{hirono2015jarzynski,vinkler2016analogy,ito2016backward,dinis2020phase}. In
particular, Vinkler {\em et al.} \cite{vinkler2016analogy} derive a
physical analogy of the gambler's problem which is  similar to the
one we shall discuss below for the GM model. 
{\red{In this analogy, the heat bath has no direct counterpart, so the temperature is arbitrary. 
As we shall see, the analogy of the GM extends further, because we identify $T$ in terms 
of the parameters of the model.} Taken together, these
results suggests a broader range of validity of inequalities such as
(\ref{Gineq}), that bound the gain that informed traders can achieve.

In the rest of the paper, we first define the GM model in Section~\ref{sec:GM}. 
In Section \ref{sec:ITGM} we discuss its connection with a Szil\'ard information engine 
and we derive the inequality (\ref{ineqW}). There we also discuss the gain of informed 
traders and present the main result (Eq.~\eqref{Gineq}).  The analytic proofs of the 
inequality~\eqref{Gineq} are presented in Section~\ref{sec:proofs}, together with a 
discussion of the analytic properties of the expected gain. We conclude with a general 
discussion of our findings and of their possible extensions in Section~\ref{sec:conclusions}. 

We should perhaps add a word about the authorship of the various parts of this paper, since 
they are quite disparate in both subject matter and style. The material relating to financial 
markets, stochastic thermodynamics and information theory (Sections~\ref{sec:GM},~ \ref{sec:ITGM} 
and~\ref{sec:conclusions}) was written by the first two authors only, but with the asymptotic 
statement~\eqref{WeakAsymp} and the inequality~\eqref{Gineq} originally being conjectural.  Analytic proofs 
of these two results were then found by the third author, and Section~\ref{sec:proofs}, which
contains these proofs and some other related material, is due to him.

\section{The Glosten-Milgrom model}
\label{sec:GM}

We consider a simplified version of the GM
model~\cite{glosten1985bid}, that describes a population of traders,
who buy and sell a stock from a dealer, whom we shall call the {\em
  market maker}. The stock has a value $Y$ which is either one, with
probability $P(Y=1)=p$, or zero (with probability $1-p$). The value of
$Y$ is known to a fraction $\nu$ of the traders -- the informed
traders -- and is unknown to both the market maker and the remaining
fraction ($1-\nu$) of the traders -- the noise traders.  
The amount of information that informed traders have is quantified by the 
entropy\footnote{\label{notations}We shall use capital letters for random variables and
  the corresponding lowercase letter for their {realised} values. At time
  $t$, the sequence of realised transactions up to that time are
  known, whereas the transactions for later times are unknown. So we
  shall use $x_{\tau}$ for $\tau\le t$ and $X_{\tau}$ for
  $\tau>t$. The expected value of a random variable $X$ with
  distribution $P(X=x)=p(x)$ is denoted as
  $\mathbb{E}[X]=\sum_x xp(x)$. The entropy of $X$ is given by the
  standard formula
  $H[X]=-\sum_x p(x)\log p(x)=-\mathbb{E}[\log p(X)]$. Likewise the
  mutual information between random variables $X$ and $Y$ is given by
  $I(X,Y)=\sum_{x,y}p(x,y)\log\frac{ p(x,y)}{p(x)p(y)}$, in terms of
  the joint distribution $p(x,y)$ and the marginals $p(x), p(y)$. 
  The mutual information is also expressed in terms of the Kullback-Leibler divergence 
$I(X,Y)=D_{KL}\left(P(X,Y)||P(X)P(Y)\right)=\mathbb{E}\left[\log\frac{P(X,Y)}{P(X)P(Y)}\right]$. 
We shall use $D_{KL}\left(P(Y|x)||P(Y)\right)$ to indicate the Kullback-Leibler divergence 
between the distribution of $Y$ condifional to $X=x$ and the unconditional distribution. We
  measure information in nats, using natural logarithms.}
\begin{equation}
\label{eq:HY}
H[Y]\=\mathbb{E}[-\log P(Y)]\=-p\log p-(1-p)\log (1-p)\,.
\end{equation}

Trading occurs sequentially at discrete times, $t=0,1,2,\ldots$. At each time, one
trader is randomly drawn from the population and submits an order,
either to buy ($X_t=1$) or to sell ($X_t=0$) one unit of the stock.  An
informed trader will buy the stock if $Y=1$ and will sell it if
$Y=0$. A noise trader will buy or sell with equal probability. The
probability that the market maker will receive a buy ($X_t=1$) or a
sell ($X_t=0$) order can be written in the compact form
\begin{equation}
  \label{PxY}
  P(X_t=x|Y=y)\=\frac{1-\nu}{2}+\nu\delta_{x,y}.
\end{equation}

The market maker doesn't know the value of $Y$, nor whether
she\footnote{\red{We follow Osborne's suggestion (see preface in \cite{osborne1994course}) on the gender of players in game theory, referring to the market maker as female and to traders as males.}} is dealing with an informed or an uninformed
trader. Yet she knows the probabilities $p$ and $\nu$.  Before observing 
the next order $X_{t+1}$, the market
maker announces an ask price $a_{t+1}$, which is the price at which 
she will sell, and a bid price
$b_{t+1}$, which is the price at which she will buy. These
prices are set in order to ensure that her profit is zero, on
average. This condition is standard in competitive markets. A market
maker with negative expected profit will be exploited by traders, and
if the expected profit is positive, she will be outcompeted by other
market makers. Since trading is a zero-sum game, this condition arises
from a minimax principle: traders submit their orders to the market
maker that make minimal profit (i.e. minimal loss to them). Given
this, market makers will try to maximise their profits.  The profit of
the market maker for selling (buying) the stock at price \red{$a_{t+1}$ at time $t+1$ is
$a_{t+1}-Y$ (respectively $Y-b_{t+1}$).  The equations for $a_{t+1}$ and $b_{t+1}$ are
a consequence of the zero expected} profit condition, where the expectation on
$Y$ is taken conditional on the information that the dealer has up to
time $t$. This leads to

\begin{align}
  a_{t+1}&\=\mathbb{E}[Y|x_{\le t},X_{t+1}=1]\label{ask}\\
  b_{t+1}&\=\mathbb{E}[Y|x_{\le t},X_{t+1}=0],
\end{align}
where $x_{\le t}=(x_0,x_1,\ldots,x_{t})$ is the observed history of
transactions up to time $t$. The price of the realised transaction
will be $p_{t+1}=a_{t+1}$ if $X_{t+1}=1$ and $p_{t+1}=b_{t+1}$
otherwise.

\begin{figure}[ht]
  \centering \includegraphics[width=0.6\textwidth,angle=0]{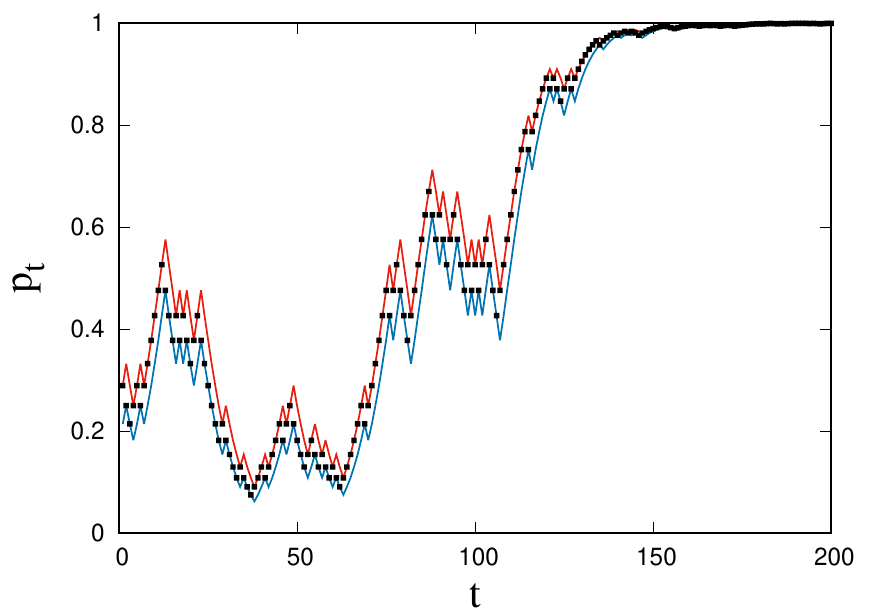}
  \caption{\label{fig:trajectories} Sample trajectories of the price
    $p_t$ ($\smallblacksquare$) for $Y=1$, 
    \red{$p=1/4$ and $\nu=0.1$. The values of $a_t$ and $b_t$ are {shown}
    as full red and blue lines}, respectively.}
\end{figure}

As time goes on, the dealer will acquire more and more information on
the true value $Y$ of the stock, from the sequence of trades
$x_{\le t}$. As a consequence, both bid and ask prices will converge
to $Y$ as $t\to\infty$. In order to see this, let us derive explicit
expressions for $a_{t+1}$ and $b_{t+1}$. We start {by} using {Bayes's} formula
to evaluate Eq.~(\ref{ask}):
\begin{eqnarray}
  a_{t+1} & = & P(Y=1|x_{\le t},X_{t+1}=1) \\
          & = & \frac{P(x_{\le t},X_{t+1}=1|Y=1)P(Y=1)}{P(x_{\le t},X_{t+1}=1)}\\
          & = & \biggl[1+\frac{1-p}{p}\Bigl(\frac{1+\nu}{1-\nu}\Bigr)^{t-2n_t-1}\biggr]^{-1} \;.
                \label{askexp}
\end{eqnarray}
Here we used $P(Y=1)=p$ and the fact that, according to
Eq.~(\ref{PxY}), $x_t$ are independently drawn conditional on~$Y$. Hence
\[
  P(x_{\le t},x_{t+1}=1|Y=1)\=\Bigl(\frac{1+\nu}{2}\Bigr)^{n_t+1}\,\Bigl(\frac{1-\nu}{2}\Bigr)^{t-n_t},
  \qquad   n_t\={\sum_{t'=0}^t x_{t'}}\,.
\]
As a consequence $a_t$ depends only on the number $n_t$ of buy trades
up to $t$, not on the order in which they occurred.  Similarly, we find
\begin{equation}
  \label{bidexp}
  b_{t+1}\= \biggl[1+\frac{1-p}{p}\Bigl(\frac{1+\nu}{1-\nu}\Bigr)^{t-2n_t+1}\biggr]^{-1}.
\end{equation}
The dealer will always sell at a price $a_t$ higher than the one
($b_t$) at which she buys. The difference $a_t-b_t\ge 0$ is called the
{bid-ask spread} in finance \cite{bouchaud2018trades}.

Notice that $n_t$ follows a binomial distribution
\begin{equation}
  \label{binomnt}
  P(n_t|Y)\=\binom t{n_t}\,\;\Bigl(\frac{1-\nu}{2}+\nu Y\Bigr)^{n_t}\,\Bigl(\frac{1+\nu}{2}-\nu Y\Bigr)^{t-n_t}\,,
\end{equation}
that depends on $Y$. The sequence $x_{\le t}$ of trades will reveal
which of the two distributions {\hbox{$P(n_t|Y=0)$}} or $P(n_t|Y=1)$ is realised.

If $Y=1$, the number of buy {transactions} will almost surely grow {like}
$n_t\simeq t (1+\nu)/2$. Likewise, if $Y=0$ we find
$n_t\simeq t (1-\nu)/2$. Eqs. (\ref{askexp},\ref{bidexp}) then show
that, almost {surely,}\footnote{A more refined analysis shows that
$\alpha_t=\log\left (1/a_t-1\right)$ performs a random walk with
  drift $\nu(1-2Y)\log\frac{1+\nu}{1-\nu}$. Hence for $Y=1$,
  $\alpha_t\to -\infty$ and $\alpha_t\to \infty$ for $Y=0$, as
  $t\to\infty$. Therefore $a_t=\left(1+e^{\alpha_t}\right)^{-1} \to Y$ as $t\to\infty$. 
  Similarly $b_t$
  can be expressed in terms of the same random walk process $\alpha_t$, but
  with different initial conditions.}
\begin{equation}
  \lim_{t\to\infty} a_t\=\lim_{t\to\infty} b_t\=Y\,.
\end{equation}
This implies that the price will ultimately aggregate all information
conveyed by the activity of informed traders, and converge to the true
value of the stock. Fig. \ref{fig:trajectories} show sample
trajectories for $a_t$ and $b_t$ in both cases $Y=0$ and $Y=1$.

The realised price $p_t$ at time $t$, i.e. the price at which the
transaction occurs, can be written as
\begin{equation}
  p_t\=\mathbb{E}[Y|x_{\le t}]. 
\end{equation}
This coincides with $a_t$ if $x_t=1$ and with $b_t$ if the last trader
was a seller ($x_t=0$). The price satisfies the martingale
property\footnote{The price $p_{t+1}=\mathbb{E}[Y|x_{\le t},X_{t+1}]$
  at time $t$ is a random variable, because it depends on the
  realisation $X_{t+1}$ of the next trade.
  $\mathbb{E}[p_{t+1}|x_{\le t}]=\mathbb{E}[\mathbb{E}[Y|x_{\le t},X_{t+1}]]$ 
  is the expected value of this random variable on~$X_{t+1}$. By the
  property of expectations, this equals~$\mathbb{E}[Y|x_{\le t}]=p_t$.}
$\mathbb{E}[p_{t+1}|x_{\le t}]=p_t$~\cite{glosten1985bid}.

\section{The Information Thermodynamics of the GM model}
\label{sec:ITGM}

The protocol that the market maker follows generates a stochastic
price process. This process is driven by the information $x_t$ on the
type of orders that the market maker receives, that leads her to
discover the true value of $Y$. In this section we map this process to
a Szil\'ard engine and we show that optimal work extraction coincides
with the optimal policy that the market maker implements.

\subsection{The Szil\'ard engine}

The Szil\'ard engine is the simplest realisation of the Maxwell's
demon idea, that the second law of thermodynamics can be violated if
some information on the microscopic state of a system is available. In
its original form, it consists of a box of unit lateral length in each
direction, containing a single point-like particle. The box is in
contact with a heat bath at temperature $T$. At a certain time $t$, a
measurement $Y$ is taken that reveals whether the particle is in the
left ($Y=0$) or the right ($Y=1$) side of the box. This allows the
observer to extract work from this system, using the following
protocol. A wall is inserted without friction and instantaneously
right after the measurement. The measurement reveals on which side the
particle is confined. Hence the observer knows on which side the
particle will push the wall. The force exerted on the wall by the
particle can be exploited to do work, e.g. to lift a weight. There are
several details in this idealised system that have been discussed at length
elsewhere~\cite{szilard-german,parrondo2001szilard,kim2011information}. In
the present context, let it suffice to say that in the limit of a
quasi-static process (i.e. when the wall moves infinitely slowly), the
work that can be extracted can be computed considering the expansion
of an ideal gas composed of a single particle. This is characterised
by an equation of state $PV=Nk_BT$, where $N$ is the number of
particles and $k_B$ is Boltzmann constant, that we shall take equal to
one ($k_B=1$) in what follows. After the wall is inserted, the
pressure on the side containing the particle is $P=T/V$, since
$N=1$. On the other side the pressure is zero, because $N=0$. In an
isothermal expansion, the work done for a change $dV$ in the volume is
$dW=PdV$. Integrating this from the initial volume $V_i=1/2$ to the
final volume $V_f=1$ yields,
\begin{equation}
  W\=\int_{1/2}^1\frac{dV}{V}\=T\log 2.
\end{equation}
This protocol leaves the system in the same state as the one before
the measurement. So the whole process constitutes a cycle whose
outcome is to extract $\log 2$ units of work from the system, in
apparent\footnote{Later Landauer~\cite{landauer1961irreversibility}
  showed that, considering the cost of storing the measurement in a
  memory, no violation of the second law of thermodynamics occurs.}
violation of the second law of thermodynamics. It has recently been
realised \cite{sagawa2012thermodynamics,parrondo2015thermodynamics} 
that, measuring a quantity $M$ during an isothermal cyclic transformation, 
allows an observer to extract, on average, an amount of work
\begin{equation}
  \label{gen2ndlaw}
  \mathbb{E}[W]\;\le\; TI(Y,M)
\end{equation}
that cannot exceed $T$ times the mutual information $I(Y,M)$ between
the the measurement $M$ and the microscopic state $Y$ of the
system. Eq.~(\ref{gen2ndlaw}) generalises the second law of
thermodynamics, for processes that use information acquired by a 
measurement system. Note that $I(Y,M)\,\le\,H[Y]$, with equality holding when 
the measurement $M$ fully determines the state $Y$ of the system.

\subsection{An analogy with the Szil\'ard engine} \label{SzilardAnalogy}

The relation of the GM model to stochastic thermodynamics is sketched
in Fig. \ref{fig1}. The variable $Y$ describes the microscopic state
of a particle in a box of lateral size $L_x=1$ at constant temperature
$T$. At time $t=0$, a wall is introduced in the box, separating it in
two parts. If $Y=0$ the particle is in the left side of the box,
whereas if $Y=1$ the particle is on the right of the partition.  At
each consecutive transaction $t=1,2,\ldots$, the market maker receives
a noisy signal $x_t$. Using this, she can operate a feedback protocol
by moving the wall, in order to extract work\footnote{We adopt the
  convention that extracted work (i.e. work done by the system) is
  positive and work done on the system is negative.} $W$. As we shall
see, the optimal work extraction protocol coincides with the market
maker pricing strategy in the GM model.

\begin{figure}[ht]
  \centering
  \includegraphics[width=0.6\textwidth,angle=0]{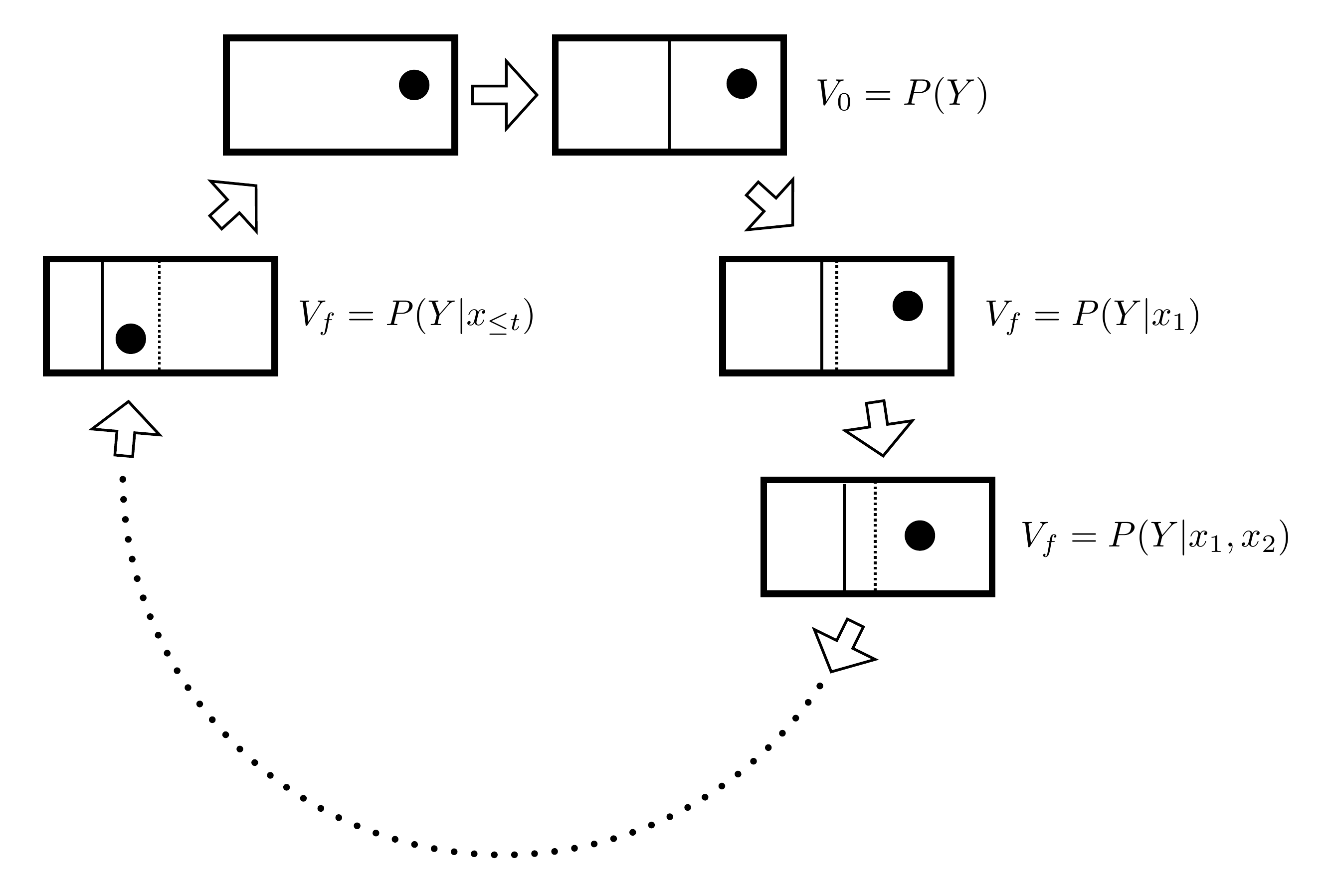}
  \caption{\label{fig1} Sequential expansion in a Szil\'ard
    engine. Initially the particle is in equilibrium in the box. At
    $t=0$, a wall is inserted at position $1-p$, and the particle is
    either on the left ($Y=0$) or on the right ($Y=1$) of the wall. At
    $t=1$, a noisy measurement $x_1$ is made of the position of the
    particle and the wall is moved from the initial position (dashed
    line) to $P(Y=0|x_1)$, therefore extracting work $W_1$. After the
    second measurement $x_2$ the wall is moved to
    $P(Y=0|x_1,x_2)$. The process continues for $t$ steps, when the
    wall is moved to $P(Y=0|x_1,\ldots,x_t)$. Then the wall is removed
    and the system returns to the initial equilibrium state.}
\end{figure}

\medskip

We first address the issue of identifying the temperature $T$. The box
containing the particle interacts with the system generating the
signals $x_t$. Both have to be at the same temperature~$T$. The
signals $x_t$ are drawn from the distribution Eq.~(\ref{PxY})
independently, conditional on~$Y$. The signals $x_{\le t}$ can be
described as a physical system of $t$ non-interacting particles, \red{in 
equilibrium at temperature $T$. Each particle can be either in the ``right''
state $x_t=Y$, or in the ``wrong'' state $x_t=1-Y$ and the energy of 
configuration $x_{\le t}$ is the sum 
\begin{equation}
  \label{Em}
  E\{x_{\le t}|Y\}\=\sum_{\tau\le t}\left[\epsilon_w\delta_{x_\tau,1-Y}+\epsilon_r\delta_{x_\tau,Y}\right]  
  \=(t-m_t)\epsilon_w+m_t\epsilon_r
\end{equation}
of the energies of the individual particles, which equal
$\epsilon_w$ or $\epsilon_r$ depending on whether $x_\tau\not =Y$ or
$x_\tau=Y$, respectively. 
In Eq.~\eqref{Em} $m_t=\sum_{\tau\le t}\delta_{x_\tau,Y}$ is the number of
particles in the ``right'' state ($x_\tau=Y$).} 
\red{In order to define the temperature $T$, we rewrite the probability of a 
micro-state $x_{\le t}$ as a Gibbs-Boltzmann distribution
\begin{equation}
  \label{Pboltz}
  P(x_{\le t}|Y) \= \Bigl(\frac{1+\nu}{2}\Bigr)^{m_t}\,\Bigl(\frac{1-\nu}{2}\Bigr)^{t-m_t} 
  \= \frac1Z\,e^{-E\{x_{\le t}|Y\}/T}\,,
\end{equation}
where $Z$ is the normalisation constant. The thermodynamics of the 
system of signals $x_{\le t}$ is defined by the free energy
\begin{equation}
\label{eqFree}
\mathcal{F}\=-T\log Z\=\mathbb{E}[E]-TH[X_{\le t}|Y]\,,
\end{equation}
where 
\[
  H[X_{\le
    t}|Y]\=-\mathbb{E}\left[\log P(X_{\le t}|Y)\right]\=t\,\Bigl[-\,\frac{1+\nu}{2}\log\frac{1+\nu}{2}\;-\;\frac{1-\nu}{2}\log\frac{1-\nu}{2}\Bigr]
\]
is the entropy of the sequence $X_{\le t}$, conditional on $Y$. 
We require that all thermodynamic effects related to variations in $\nu$ should 
be ascribed to variations in the entropic term, i.e. that the average energy 
\[
\mathbb{E}[E]\=t\left(\frac{1+\nu}{2}\epsilon_r+\frac{1-\nu}{2}\epsilon_w\right)
\]
should be independent of $\nu$\footnote{\red{A variation of $\nu$ would correspond to a thermodynamic transformation where some work $W$ is done on the system. By the first law of thermodynamics $\Delta\mathbb{E}[E]=W - Q$, the work is related to the change in the internal energy and to the heat $Q$ absorbed by the system. Assuming that $\mathbb{E}[E]$ is independent of $\nu$ implies that all work done on the system by changing $\nu$ is turned into heat $Q$, which then results in a change of the temperature $T$.}}. Without loss of generality, we set the zero of the energy
so that $\mathbb{E}[E]=0$. This allows us to express $\epsilon_r$ in terms of $\epsilon_w$ 
and it yields $Z=e^{H[X_{\le t}|Y]}$ (see Eq.~\ref{eqFree}). These relations, combined
with the expression for $H[X_{\le t}|Y]$, turn Eq.~\eqref{Pboltz} into an identity for arbitrary 
values of $t$ and $m_t$, provided that $T$ satisfies the equation  
\[
\frac{\epsilon_w}{T} \=\frac{1+\nu}{2}\log\frac{1+\nu}{1-\nu}\,.
\]
Here $\epsilon_w$ remains as the only scale of the energy. Without loss of generality, 
we set $\epsilon_w=1$, which is equivalent to measuring $T$ and $W$ in units of 
$\epsilon_w$. With this choice, the temperature takes the value}
\begin{equation}
  \label{T}
  \boxed{\phantom{\Biggl)}T\=\Bigl(\frac{1+\nu}{2}\log\frac{1+\nu}{1-\nu}\Bigr)^{-1}\,.\phantom{\Biggr)}}
\end{equation}

This definition of the temperature in Eq.~(\ref{T}) coincides with the
micro-canonical one
\begin{equation}
  \label{Tmicro}
  \frac{\partial S}{\partial E}\= \frac{1}{T}\,,
\end{equation}
in the thermodynamic limit $t\to\infty$, as it should be, by the
equivalence of ensembles. The entropy $S$ in Eq.~\eqref{Tmicro} is the logarithm of the number
of configurations with a certain energy~$E$, i.e.
\[
  S \= \log\, \binom t{m_t} \simeq \;
  t\left[-\frac{m_t}{t}\log\frac{m_t}{t}-\left(1-\frac{m_t}{t}\right)\log\left(1-\frac{m_t}{t}\right)\right]
\]
where the last relation holds in the limit $t\to\infty$. $S$ as a
function of $E$ is obtained using Eq.~(\ref{Em}) to express $m_t$ in
terms of $E$,
i.e. \red{$m_t=\frac{t\epsilon_w-E}{\epsilon_w-\epsilon_r}$. Its derivative evaluated at the typical value of the
energy ($E=\mathbb{E}[E]=0$), with $\epsilon_w=1$,} yields Eq.~(\ref{Tmicro}),
with $T$ given by Eq.~(\ref{T}).

The market temperature $T$ depends only on $\nu$, i.e. on the fraction
of informed traders. It is a decreasing function of $\nu$. It diverges
when $\nu\to 0$, i.e. when the market is dominated by noise
traders. In the absence of noise traders (i.e. when $\nu\to 1$) the
temperature $T$ tends to zero.

\medskip
Let us now show that the protocol of optimal work 
extraction in a Szil\'ard engine operated at temperature $T$ coincides with 
the market maker strategy. Our derivation follows steps similar to
Ref.~\cite{vinkler2016analogy}.  The following work extraction
protocol is operated by an agent that observes the sequence of signals
$x_1,x_2,\ldots$ that the market maker receives. We identify this
agent with the market maker for simplicity, but that is not
necessary. The main aim of the present section is to derive an upper
bound on the work that these signals allows an observer to extract
from a Szil\'ard engine.

Let us discuss the first step $t=1$.  Let the wall divide the box in
two partitions of volumes $V_0(0)$ on the left and $V_0(1)=1-V_0(0)$
on the right of the wall. In this way, $V_0(Y)$ is the volume of the
partition containing the particle. The distribution of the particle's
position in the box is uniform, so the probability that the particle
is initially found in $Y$ is equal to $V_0(Y)$. Equivalence with the
GM model implies $V_0(1)=P(Y=1)=p$.

The market maker then observes the value of $x_1$ and she can moves
the partition quasi-statically. Let $V_f^{(0)}(Y|x_1)$ be the volume
of the partition containing the particle after this change. The work
extracted for $t=1$ is
\begin{equation}
  W_1(Y)\=T\log\frac{V_f^{(0)}(Y|x_1)}{V_0(Y)}.
\end{equation}
Note that $W_1$ is a random variable, because it depends on the
(unknown) position of the {particle~$Y$}. The expected value of $W_1$,
conditional on $X_1=x_1$ is
\[
  \mathbb{E}[W_1|X_1=x_1]\=T\sum_y
  P(y|x_1)\log\frac{V_f^{(0)}(y|x_1)}{V_0(y)}{\,,}
\]
which is maximal for
\[
  V_f^{(0)}(y|x_1)\=P(y|x_1)\,.
\]
So the maximum value of the expected work extracted at $t=1$ is
\[
  \max \mathbb{E}[W_1|X_1=x_1]\=T\,I(Y,x_1){\,,}
\]
where $I(Y,x_1)= D_{KL}\bigl(P(Y|x_1)||P(Y)\bigr)$ is the information
gained on $Y$ from the measurement~$x_1$. The expected value of
$I(Y,x_1)$ on~$x_1$ is the mutual information
$I(Y,X_1)=\sum_{y,x_1} p(y,x_1)\log\frac{p(y|x_1)}{p(y)}$.

After the first step of quasi-static expansion, without removing the
wall, we find ourselves in a situation very similar to the initial
one. The new initial volume of the partition containing the particle
is $V_0^{(1)}(Y)=V_f^{(0)}(Y|x_1)=P(Y|x_1)$, which is the probability that
the particle is in partition $Y$ of the box, given the information
$x_1$. At $t=2$ the market maker receives the signal $x_2$, and
performs a quasi-static expansion to volume
$V_f^{(1)}(Y|x_2)$. Maximising the expected work that can be
extracted, we find $V_f^{(1)}(Y|x_2)=P(Y|x_1,x_2)$.  The maximum value
of the expected work done for $t=2$ is
\[
  \max \mathbb{E}[W_2|x_1,x_2] \=T\sum_y P(y|x_1,x_2)\log\frac{P(y|x_1,x_2)}{P(y|x_1)}
  \=T\,I(Y,x_2|x_1).
\]

This argument can be extended for all $t\ge 1$. Given all the
information $x_{\le t-1}=(x_1,\ldots,x_{t-1})$ received before $t$,
the optimal policy implies that $V_0^{(t-1)}(y)=P(y|x_{\le t-1})$. The
work extracted is
\[
  W_t(Y)\=T\log\frac{V_f^{(t-1)}(Y|x_t)}{P(Y|x_{\le t-1})}\,,
\]
whose expected value is maximal for
$V_f^{(t-1)}(y|x_t)=P(y|x_{\le t})$.  The expected value of the work
extracted at time $t$ is
$\max \mathbb{E}[W_t|x_{\le t}]=TI\left(Y,x_t|x_{\le t-1}\right)$.
The total work extracted up to time $t$ following this protocol is
\begin{eqnarray}
  W_{\le t} (Y)& = & W_1(Y)+W_2(Y)+\ldots, W_t(Y)
  \\
               & = & T\log\frac{P(Y|x_{\le t})}{P(Y)}.
\end{eqnarray}
so its expected value over $Y$ is
$\max \mathbb{E}[W_{\le t}|x_{\le t}]=I\left(Y,x_{\le t}\right)$.
Taking the expected value over the realisation of the trading process
$x_{\le t}$, we find
\begin{equation}
  \label{ineqmm}
  \max \mathbb{E}[W_{\le t}]\=T\,I\left(Y,X_{\le t}\right)\;\le\; TH[Y]\,,
\end{equation}
where the last inequality {holds because}
$I\left(Y,X_{\le t}\right)=H[Y]-H[Y|X_{\le t}]$ and~\hbox{$H[Y|X_{\le t}]\ge 0$}.  
If the wall is removed at time $t$, the system
reverts back to the original equilibrium state, closing the cycle.
The inequality (\ref{ineqmm}) holds as an equality in the limit $t\to\infty$, i.e.
\begin{equation}
  \label{limitW}
  \boxed{\phantom{\biggl[}\lim_{t\to\infty} \max \mathbb{E}[W_{\le t}]\=TH[Y],\phantom{\biggr]}}
\end{equation}
because $H[Y|X_{\le t}]\to 0$ as $t\to\infty$. Indeed,
$P(y|x_{\le t})$ converges to a singleton in either $y=0$ or $y=1$,
for all typical realisations of $x_{\le t}$. 

Finally, notice that the price
\begin{equation}
  p_t\=V_0^{(t)}(1|x_{\le t})
\end{equation}
is given by the volume of the right partitions, along the
process. Therefore the feedback protocol that extracts the maximal
amount of work coincides with the optimal behaviour of the market
maker.

\subsection{The Gain of Informed Traders}

In this section we discuss the profits and losses of different market 
participants.
As discussed in Section~\ref{sec:GM}, the market-maker sets the prices 
in order to have zero expected profit at each time step. 
Since the market is a zero sum game, the expected total gain of
informed traders equals the expected total loss of noise traders. 
Therefore it is sufficient to focus attention on the gain of informed traders,
which is what we shall do in what follows.

Let $G_t$ denote the total expected gain of informed traders up to
time $t$. This satisfies the recursion relation
\begin{equation}
  \label{recG}
  G_{t+1}\=U_{t+1}\bigl[b_{t+1}(1-Y)+(1-a_{t+1})Y\bigr]\+G_t\,,
\end{equation}
where $U_{t}$ is a random variable that takes value $U_t=1$ if the
trader at time $t$ is informed and $U_t=0$ otherwise. Indeed, when
$Y=0$ an informed trader will gain $b_{t+1}$ by selling the stock,
whereas if $Y=1$ he will buy at price $a_{t+1}$ realising a gain
$1-a_{t+1}$. We are interested in the total asymptotic gain, defined as
\begin{equation}
  \label{eqG}
  G\=\lim_{t\to\infty} G_t\=\sum_{t=0}^\infty U_{t+1}\bigl[b_{t+1}(1-Y)+(1-a_{t+1})Y\bigr]\,,
\end{equation}
where the last equality derives from the recursion relation~\eqref{recG} and the initial condition $G_0=0$.
The expected value of $G$ can be computed taking conditional expectations
\begin{equation}
  \label{eqGp}
  \mathbb{E}[G]\=p\,\mathbb{E}[G|Y=1]\+(1-p)\,\mathbb{E}[G|Y=0].
\end{equation}
Taking the expectation of Eq.~(\ref{eqG}) conditional to $Y=0$, we find
\begin{eqnarray}
  \mathbb{E}[G|Y=0] & = & \nu\,\sum_{t=0}^\infty \mathbb{E}\left[ b_{t+1}\right] \nonumber \\
                    & = & \nu\,\sum_{t=0}^{\infty} \sum_{k=0}^t \,\binom tk\,
 \frac{\bigl(\frac{1-\nu}{2}\bigr)^k\,\bigl(\frac{1+\nu}{2}\bigr)^{t-k}}
      {1+\frac{1-p}{p}\bigl(\frac{1+\nu}{1-\nu}\bigr)^{t-2k+1}}\,.
                          \label{eqG0}
\end{eqnarray}
In the first line we used the fact that the variables $U_{t}$ are
independent and identically distributed, with
$\mathbb{E}[U_{t+1}]=\nu$. Eq.~(\ref{eqG0}) follows from the
expression Eq.~(\ref{bidexp}) for $b_{t+1}$ and Eq.~(\ref{binomnt})
for the binomial distribution over which the expectation is taken.
The expectation, conditional to $Y=1$ is computed in the same way,
substituting $b_{t+1}$ with $1-a_{t+1}$:
\begin{eqnarray}
  \mathbb{E}[G|Y=1]&=&\nu\,\sum_{t=0}^{\infty} \sum_{k=0}^t \,\binom tk\;
   \Bigl(\frac{1+\nu}{2}\Bigr)^k\,\Bigl(\frac{1-\nu}{2}\Bigr)^{t-k} 
   \Biggl(1-\frac{1}{1+\frac{1-p}{p}\bigl(\frac{1+\nu}{1-\nu}\bigr)^{t-2k-1}}\Biggr) \nonumber \\
                   &=&\nu\,\sum_{t=0}^{\infty} \sum_{k=0}^t \binom tk\,
    \frac{\bigl(\frac{1-\nu}{2}\bigr)^k\bigl(\frac{1+\nu}{2}\bigr)^{t-k}}
      {1+\frac{p}{1-p}\bigl(\frac{1+\nu}{1-\nu}\bigr)^{t-2k+1}}\,.
                       \label{eqG1}
\end{eqnarray}
Here we have used the change of index $k\rightarrow t-k$ in the last
step. Notice that $\mathbb{E}[G|Y=1]$ in Eq.~(\ref{eqG1}) is equal to
Eq.~(\ref{eqG0}) {with $p$ replaced by~$1-p$}. This implies
that $\mathbb{E}[G]$ is invariant under the transformation $p\to 1-p$,
as it should because the GM model enjoys the same invariance. 

From equations~\eqref{eqG0} and~\eqref{eqG1} it is clear that both
$\mathbb E[G|Y=0]$ and $\mathbb E[G|Y=1]$, and hence also their sum $\mathbb E[G]$,
are decreasing functions of~$\nu$. Also note that $\mathbb{E}[G]\to 0$ as $\nu\to 1$,
as it should be, since if all traders are informed, the market maker can guess the 
value of $Y$ from the first transaction.

\begin{figure}[ht]
  \centering \includegraphics[width=0.9\textwidth,angle=0]{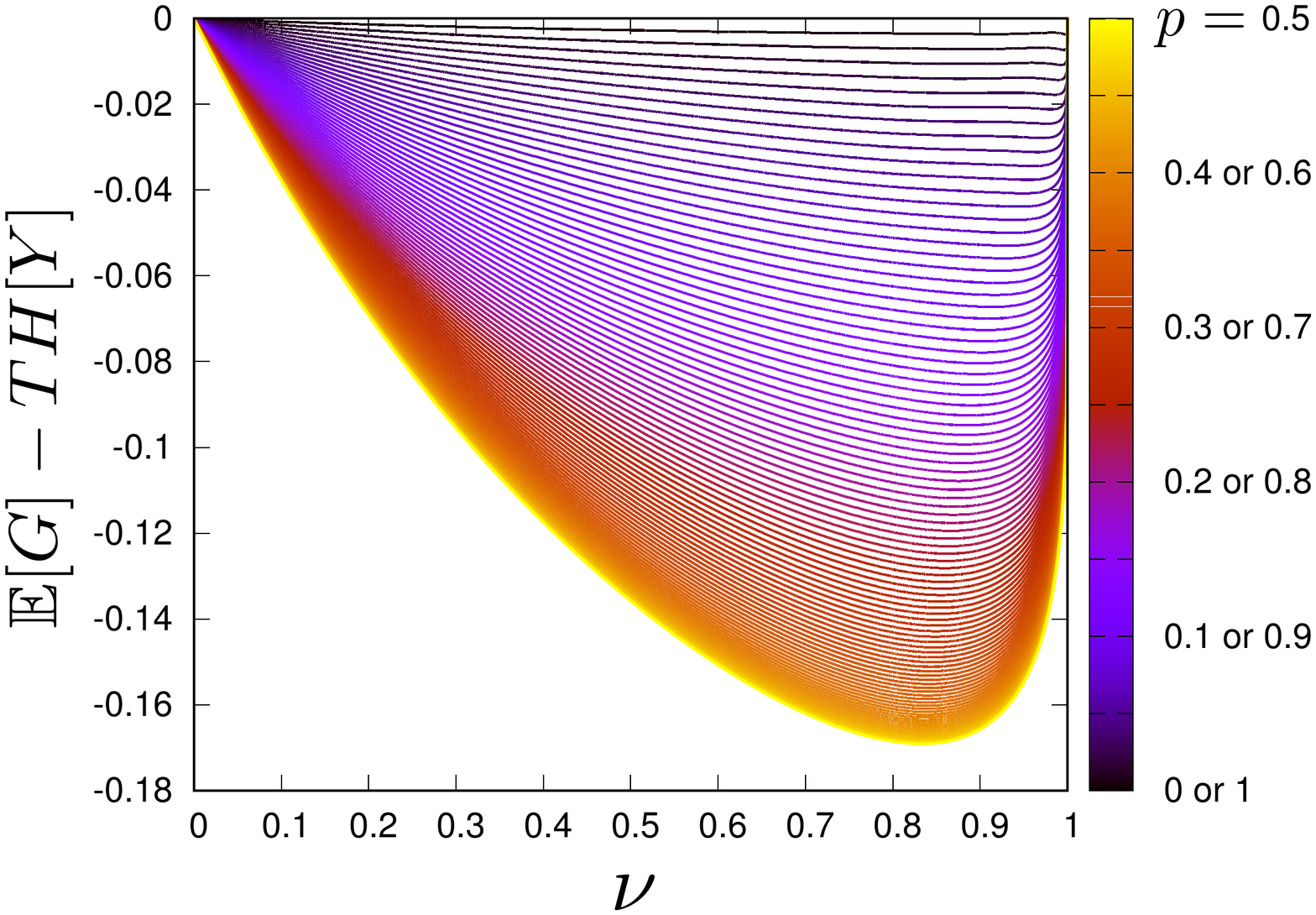}
  \caption{\label{fig:bound} Numerical evaluation of the function
    $\mathbb{E}[G]-TH[Y]$ as a function of $\nu$, for different values
    of $p$, on a scale defined by the color-code on the right.}
\end{figure}

{The expression for $\mathbb{E}[G]$ that we obtain by combining 
Eqs.~\eqref{eqGp},~\eqref{eqG0} and~\eqref{eqG1} has no evident analogue in the 
physical system. Yet numerical analysis gives strong indications in support of the 
inequality~\eqref{Gineq}, as shown in Fig.~\ref{fig:bound}. This also suggests
the limiting behavior
\begin{equation} \label{WeakAsymp}
  \lim_{\nu\to 0} \frac{\mathbb{E}[G]}{T}\=\lim_{T\to \infty} \frac{\mathbb{E}[G]}{T}\=H[Y]
\end{equation}
as $\nu\to0$ or $T\to\infty$ (recall that $\nu$ and~$T$ are related by~\eqref T).
In the next section we will give proofs of the inequality~(\ref{Gineq}) and of Eq.~(\ref{WeakAsymp}). 
In fact, in Section~\ref{sec:proofs} we will derive 
a refinement of~\eqref{WeakAsymp} to a complete asymptotic expression
\be\label{StrongAsymp}
  \mathbb{E}[G] \;\simeq \; H[Y]\,T \,-\, \frac {5p(1-p)}{3T} \+ \frac {10p(1-p)}{3T^2} 
  \,-\,\frac{2p(1-p)(173-3p+3p^2)}{45T^3} \+\cdots 
\ee
for $\mathbb{E}[G]$ as a Laurent series in~$1/T$ as $T\to\infty$ (or~$\nu\to 0$), giving a quantitative
version of the inequality~\eqref{Gineq} in the large temperature limit. The coefficients in  the 
expansion~\eqref{StrongAsymp} are all polynomials in $p(1-p)$, consistently with the invariance 
under the transformation $p\to 1-p$ discussed above.} 

The inequality~\eqref{Gineq}, which is the main result of this paper, gives an upper bound 
to the gains that informed traders can extract from their side information. 
Combined with the derivation in the previous section, it shows 
that the maximal profit that informed traders can gain cannot exceed the
maximal work that can be extracted from the analogous physical system (the Szil\'ard
information engine). This reveals a precise analogy between 
the generalized second law of thermodynamics~\eqref{gen2ndlaw} and a 
generalised efficient market hypothesis, for the GM model.
The asymptotic statements~\eqref{WeakAsymp} and~\eqref{StrongAsymp}, on the other hand, says that 
for~$\nu\to 0$ the approach of the price $p_t$ to its asymptotic value $Y$ becomes infinitely slow,
resembling a quasi-static limit in thermodynamics. In this limit, the transactions of informed traders are 
well separated in time and the market has enough time to relax to an equilibrium dominated by noise 
traders, as in a reversible process. Eq.~\eqref{WeakAsymp}, hence, is consistent with the 
observation~\cite{horowitz2011} that the generalised second law of thermodynamics holds as 
an equality for reversible feedback protocols.

\section{Asymptotics and inequalities for the expected gain function}
\label{sec:proofs}

In this section we will study the properties of the functions $\mathbb{E}[G|Y=0]$, $\mathbb{E}[G|Y=1]$ 
and $\mathbb{E}[G|Y=1]$ as defined in equations~\eqref{eqG0}, \eqref{eqG1} and \eqref{eqGp}, respectively, 
and in particular prove the asymptotic development~\eqref{StrongAsymp} and inequality~\eqref{Gineq}.

It will be convenient to replace $\nu$ by an equivalent variable~$q$, again between 0 and~1, by setting
\be\label{defq}  q \= \frac{1-\nu}{1+\nu}\,, \qquad  \nu \= \frac{1-q}{1+q}\,. \ee
We can then rewrite equations~\eqref{eqG1}, \eqref{eqG0} and \eqref{eqGp} in terms of this new variable as
\begin{equation} \label{G0G1G} \mathbb{E}[G|Y\!=\!1]\,=\,G_q(p),\;\, \mathbb{E}[G|Y\!=\!0]\,=\, G_q(1-p),\;\,
      \mathbb{E}[G]\=p\,G_q(p)\+(1-p)\,G_q(1-p), \end{equation}
where $G_q(p)$ is defined by
\begin{equation} \label{defGq} 
G_q(p)\= \sum_{t=0}^{\infty}\frac{1-q}{(1+q)^{t+1}}\sum_{k=0}^t\,\binom tk\,\frac{q^k}{1+\frac p{1-p}\,q^{2k-t-1}}\;.
\end{equation}
It is therefore this function that we have to study.

The formula for $G_q(p)$ can be simplified considerably. We first split up the sum into two parts 
according to whether $k\le t/2$ or $k>t/2$. Writing $(t, k)$ as $(2m + n,m)$ in the
first case and as $(2m + n,m + n)$ in the second leads to
\begin{equation}
  G_q(p) \= \sum_{n=0}^{\infty} \frac{P_n(q)}{1+\frac{p}{1-p}\,q^{-n-1}} 
   \+ \sum_{n=1}^{\infty} \frac{q^nP_n(q)}{1+\frac p{1-p}\,q^{n-1}}
\end{equation}
with $P_n(q)$ defined by
\begin{equation}
  P_n(q) \= \frac{1-q}{(1+q)^{n+1}}\,\sum_{m=0}^{\infty}\binom{2m+n}m\;\frac{q^m}{(1+q)^{2m}}\,.
\end{equation}
But $P_n(q)=1$ by a standard identity (that can be proved, from the binomial theorem for~$n=0$ 
and then by induction on~$n$), so the formula for $G_q(p)$ simplifies to
 $$ G_q(p) \= \sum_{n=0}^{\infty} \frac{1}{1+\frac{p}{1-p}\,q^{-n-1}} 
    \+  \sum_{n=1}^{\infty} \frac{q^n}{1+\frac{p}{1-p}\,q^{n-1}} 
  \= \sum_{n\in\mathbb Z} \frac{q^{\max(0,n)}}{1+\frac{p}{1-p}\,q^{n-1}}\;.$$
This can be rewritten as
\be\label{FtoG} G_q(p)\= F\Bigl(\frac{1-p}p,q\Bigr)\+\frac{1-p}p\,q\,F\Bigl(\frac p{1-p},q\Bigr)\,-\,\frac{(1-p)(1-q)}2\;, \ee
where $F(x,q)$ is the relatively simple function defined for all $x$ and $q$ in $\mathbb C$ with $|q|<1$ by
\be\label{defF}  F(x,q) \= \frac12\,\frac x{1+x} \+ \sum_{n=1}^\infty\frac{q^nx}{1+q^nx} \;. \ee
Eq.~\eqref{G0G1G} then gives 
\be \label{eqGsimple}
  \mathbb{E}[G] \= (1+q)\,\Bigl[p\,F\Bigl(\frac{1-p}p,q\Bigr)\+(1-p)\,F\bigl(\frac p{1-p},q\bigr)\Bigr]\,-\,p(1-p)(1-q)\,.
\ee
In the remainder of this section we will give a series of properties of the function $F(x,q)$ 
and use them and Eq.~\eqref{eqGsimple} to prove~\eqref{StrongAsymp} and~\eqref{Gineq}.  

{\bf 1.} The simplest property of $F(x,q)$ is that it satisfies the functional equation
\be  F(x,q) - F(qx,q) \= \frac12\,\Bigl(\frac x{1+x}\+\frac{qx}{1+qx}\Bigr)\,, \ee
as one sees by replacing $n$ by $n+1$ in~\eqref{defF}. A consequence is that the antisymmetrized function
\be F_-(x.q) \=  F(x,q) \,-\, F(x^{-1},q) \+ \frac{\log x}{\log q} \ee
is invariant under $x\mapsto qx$ and therefore has a Fourier expansion in $\frac{\log x}{\log q}$,
and indeed by using the Poisson summation formula we find the rapidly convergent Fourier sine expansion
\be\label{sine} F_-\bigr(e^{uh},e^{-h}\bigr) 
 \= \frac{2\pi}h\,\sum_{n=1}^\infty\frac{\sin(2\pi n u)}{\sinh(2\pi^2 n/h)}\qquad(h>0)\,. \ee
We omit the proof, since we will not use this formula. We also mention in passing that $F_-$ can be expressed
in closed form in terms of the Weierstrass zeta-function $\zeta(z;\tau)$ (the function whose derivative is the 
Weierstrass $\wp$-function; we do not give the complete definition since it plays no further role in this paper) with 
$q=e^{2\pi i\tau}$ and $x=-e^{2\pi iz}$, and that Eq.~\eqref{sine} can then also be obtained as a consequence
of the transformation behavior of $\zeta(z;\tau)$ under $(z,\tau)\mapsto(z/\tau,-1/\tau)$.

{\bf 2.}  If we expand each term $\frac{q^nx}{1+q^nx}$ ($n\ge0$) in~\eqref{defF} as a geometric series
$\sum_{r=1}^\infty(-1)^{r-1}q^{nr}x^r$ and then resum the resulting geometric series in~$n$, we obtain the formula 
\be\label{Fnew} F(x,q) \= \sum_{r=1}^\infty\frac{(-1)^{r-1}}2\,\frac{1+q^r}{1-q^r}\,x^r \ee
giving the Taylor expansion of~$F$ at~$x=0$.  This formula is only valid for $|x|<1$, but if we
retain the first $N$~terms of~\eqref{defF} and expand only the others as geometric series we 
obtain the more general hybrid series expansion
\be\label{Fhyb}  F(x,q) \= \frac12\,\frac x{1+x} \+ \sum_{n=1}^{N-1}\frac{q^nx}{1+q^nx} 
 \+ \sum_{r=1}^\infty (-1)^{r-1}\,\frac{q^{Nr}}{1-q^r}\,x^r \,,\ee
which is now valid whenever $|x|<|q|^{-N}$, and hence for any~$x\in\mathbb C$ if we take $N$ large enough.
This is also useful computationally (although the convergence of either~\eqref{defF} or~\eqref{Fnew} is
already exponential and hence good enough in practice), since truncating the second sum in~\eqref{Fhyb}
after $N$ terms gives an approximation of $F(q,x)$ up to order $q^{N^2}$ in only O($N$) rather than O($N^2$) steps.

{\bf 3.}  We now consider the behavior of $F(x,q)$ near $q=1$.  Set $q=e^{-h}$ with~$h>0$. Replacing {
$\frac 1 2 \frac{1+q^r}{1-q^r}=\frac12+\frac1{e^{rh}-1}$} in equation~\eqref{Fnew} by its Laurent expansion in terms 
of Bernoulli numbers, we get the asymptotic expansion
\begin{align}\label{Near0}  \nonumber F(x,e^{-h}) 
  & \;\sim\; \sum_{r=1}^\infty\,(-1)^{r-1}\,\biggl(\frac1{rh}\+\sum_{n=2}^\infty\frac{B_n(rh)^{n-1}}{n!}\biggr)\,x^r \\
  &\= \frac{\log(1+x)}h \+ \sum_{n=2}^\infty \frac{B_n}{n!}\,\Bigl(x\,\frac d{dx}\Bigr)^{n-1}\Bigl(\frac x{1+x}\Bigr)\,h^{n-1}
  \nonumber  \\   &\= \frac{\log(1+x)}h \+ \frac{x}{(1+x)^2}\,\frac h{12} 
  \,-\, \frac{x-4x^2+x^3}{(1+x)^4}\,\frac{h^3}{720} + \cdots \end{align}
as a Laurent series in~$h$, in which the coefficient of $h^{2k-1}$ for each $k>0$ is $\frac{B_{2k}}{(2k)!}$ 
times a polynomial of degree~$k$ in $\frac x{(1+x)^2}$ with integral coefficients.  The fact that this series
is odd implies that $F_-(x,e^{-h})$ vanishes to all orders in~$h$ as $h\!\!\searrow\!0$, 
but in fact we know from~\eqref{sine} that $F_-(x,e^{-h})=\text O(e^{-2\pi^2/h})$.

{\bf 4.} Finally, we discuss upper bounds for~$F$. The expansion~\eqref{Near0} implies that the difference
$$ F_0(x,q) \;:=\; F(x,q\bigr) \,-\, \frac1h\log(1+x)  $$
is bounded by a multiple of $h$ as $h\to0$ or~$q\to1$ with $x$ fixed.  The first part of the
following proposition makes this upper bound explicit, the second part gives an explicit upper 
bound for $F_0(x,q)$ that is independent of~$h$ (and hence stronger than the first bound when $h$ is 
large), and the third part combines these two to give a uniform upper bound valid for all values 
of $x$ and~$h$.  In each case the proof would also give a lower bound, but we do not write these 
out explicitly since they are not particularly relevant for the purposes of this paper.
 \newtheorem*{Prop}{Proposition}  
 \begin{Prop}  
For $x,\,h>0$ and $q=e^{-h}$ we have the upper bounds
 \begin{align}\label{Bound1}  F_0(x,q\bigr) &\;< \; \frac h{12}\,\cdot\;
 \begin{cases} \frac x{(1+x)^2} &\text{if $\,x\le 1$,} \\ \phantom X\frac12  &\text{if $\,x\ge 1$,} \end{cases} \\
 \label{Bound2}  F_0(x,q\bigr) &\;< \; \frac12\,\frac x{1+x}
   \,-\; \frac1h\,\cdot\,\begin{cases}  \log\bigl(\frac{1+x}{1+x\,\sqrt q}\bigr) &\text{if $\,x\le q^{-1/2}$,} \\ 
   \log\bigl(\frac{1+\sqrt q}{1+q}\bigr)\vphantom{\Bigr|} &\text{if $\,x\ge q^{-1/2}$,} \end{cases}\;.  \\
 \label{Bound3}  F_0(x,q\bigr) &\;< \; \frac12\,\frac x{1+x}\,\frac{1-q}{1+q}\qquad\text{for all $x$ and $h$.}  \end{align}
 \end{Prop}
\begin{proof} We first observe that the Laurent expansion~\eqref{Near0} could also have been deduced directly 
from the definition~\eqref{defF}, rather than from the alternative formula~\eqref{Fnew}, by using 
the Euler-Maclaurin summation formula, and that the Euler-Maclaurin formula also has a finite form that gives 
an explicit upper bound for the truncation at any point.  This finite form, obtained by $K$-fold integration by parts
(see for instance Prop.~3 of~\cite{zagier2007mellin} and its proof) says that for any 
smooth function $f:[0,\infty)\to\mathbb C$ that is small at infinity and any integer~$K\ge1$ we have
$$ \frac{f(0)}2+\sum_{n=1}^\infty f(nh) \,=\, \frac1h \int_0^\infty\! f(t)\,dt \,-\, 
  \sum_{k=2}^K \frac{B_k}{k!}\,f^{(k-1)}(0)\,h^{k-1} 
  \,+\, \frac{(-h)^{K-1}}{K!}\,\int_0^\infty \overline{B_K(t/h)}\,f^{(K)}(t)\,dt $$
for all $h>0$, where $\overline{B_K(x)}=B_K(x-[x])$ is the periodic version of the $K$th Bernoulli polynomial.
Applying this with $K=2$ and $f(t)=f_x(t):=\frac{xe^{-t}}{1+xe^{-t}}$ gives
$$ F_0(x,q) \= \frac h{12}\,\frac x{(1+x)^2}
 \;-\; \frac h2\,\int_0^\infty\overline{B_2(t/h)}\;\frac{xe^{-t}(1-xe^{-t})}{(1+xe^{-t})^3}\,dt $$
and this implies our first bound~\eqref{Bound1} since 
$\bigl|\overline{B_2(x)}\bigr|\le\max\limits_{0\le x\le 1}\,\bigl|x^2-x+\frac16\bigr|=\frac16$ and
$$\int_0^\infty\biggl|\frac{xe^{-t}(1-xe^{-t})}{(1+xe^{-t})^3}\biggr|\,dt
  \= \begin{cases} \phantom X\frac x{(1+x)^2} & \text{if $0\le x\le1$,} \\
     \,\frac12 - \frac x{(1+x)^2} & \text{if $x\ge1.$} \end{cases} $$
If we took $K=1$ instead of~$K=2$ we would get $F_0(x,q)\le\frac12\,\frac x{1+x}$, and in fact for this we 
would not need the Euler-Maclaurin formula at all but simply the observation that $f_x(n)\le\int_{n-1}^nf_x(t)\,dt$
for all $n\ge1$ because $f_x$ is monotone decreasing.  To obtain the stronger inqeuality~\eqref{Bound2},
we use instead the observation that $f(n)<\int_0^{1/2}\bigl(f(n+t)+f(n-t)\bigr)\,dt$ if $f$ is convex on the 
interval~$[n-\frac12,n+\frac12]$. For $f=f_x$ this holds for all $n\ge1$ if $x\le q^{-1/2}$ and hence
$\sum_{n=1}^\infty f_x(n)\le\int_{1/2}^\infty f_x(t)\,dt = \frac1h\log\bigl(1+x\sqrt q\bigr)$,
which is equivalent to the first inequality in~\eqref{Bound2}.  If $x\ge q^{-1/2}$, then there is a unique
integer $N\ge1$ for which $q^{1/2}\,<\,x q^N \,\le\,q^{-1/2}$.  Then $f_x(t)$ is decreasing on $(0,N)$ and
convex on $\bigl(N+\frac12,\infty\bigr)$, so we get instead
$$\sum_{n=1}^\infty f_x(n)\,-\,\int_0^\infty f_x(t)\,dt \;\le\;
 -\int_N^{N+1/2} f_x(t)\,dt \= \frac1h\,\log\biggl(\frac{1+xq^N}{1+xq^{N+1/2}}\biggr)
\;\le\;\frac1h\,\log\biggl(\frac{1+q^{1/2}}{1+q}\biggr)\,, $$
proving the second inequality in~\eqref{Bound2} as well.
Finally, the inequality \eqref{Bound3}~follows from~\eqref{Bound2} if $x\le q^{-1/2}$ because the difference
$$ \delta_h(x) \,:=\;\frac1h\log\biggl(\frac{1+x}{1+x\,\sqrt q}\biggr) \,-\,\frac x{1+x}\,\frac q{1+q} $$
between their right-hand sides is~$\ge0$ for all~$x\ge0$, since $\delta_h(0)=0$ and
\begin{align*} (1+x)^2\,\delta_h^{\,\prime}(x) &\= \frac{1-\sqrt q}h \,\frac{1+x}{1+x\,\sqrt q} \,-\,\frac q{1+q} \\
  &\;\ge\; \frac{1-\sqrt q}h \,-\,\frac q{1+q} 
  \,\=\, \frac1{1+e^h}\sum_{n=2}^\infty\Bigl(1-\frac{1+(-1)^n}{2^n}\Bigr)\,\frac{h^{n-1}}{n!}\;\ge\,0\;.  \end{align*}
If $x\ge q^{-1/2}$ then the difference between the right-hand sides of~\eqref{Bound2} and~\eqref{Bound3} is 
bounded above by $\frac1{1+q}-\frac1h\log\bigl(\frac{1+\sqrt q}{1+q}\bigr)$, which is negative for $h>2.784$,
while the difference between the right-hand sides of~\eqref{Bound1} and~\eqref{Bound3} is bounded
by $\frac x{1+x}\bigl(h\frac{1+\sqrt{q}}{24}-\frac12\frac{1-q}{1+q}\bigr)$, which is negative for $h<11.969$. 
\end{proof}

We can now complete the proofs of the assertions in Section~2 of this paper by combining the formula~\eqref{eqGsimple} 
for $\mathbb E[G]$ in terms of $F(x,q)$ with the results~{\bf 1.--4.}  With the changes of variables~\eqref{defq} and 
$q=e^{-h}$ (or equivalently $\nu=\tanh(h/2)$), the market temperature defined in~\eqref T is given by
\be\label{Texp}  T \= \frac{1+q}h \= \frac{1+e^{-h}}h \= \frac2h \,-\, 1 \+ \frac h2 \,-\,\frac{h^2}6 \+ \cdots  \ee
so that $q\to1$ or $h\to0$ corresponds to large temperature, while the entropy $H[Y]$ is given by~\eqref{eq:HY}. 
Substituting the expansion~\eqref{Near0} into~\eqref{eqGsimple} therefore gives
$$ \frac{\mathbb{E}[G]}T \;\sim \; H[Y] \,-\, \frac {5p(1-p)}{12}\,h^2 
  \+\frac{p(1-p)(29+6p(1-p))}{720}\,h^4 \+\cdots\,,$$
which in view of~\eqref{Texp} is equivalent to~\eqref{StrongAsymp}, but is somewhat simpler because it is
an even power series~$h$ whereas~\eqref{StrongAsymp} is not even or odd in powers of~$1/T$.  Finally, 
substituting~\eqref{Bound3} into~\eqref{FtoG} gives
$$  (1+q)\,p\,F\Bigl(\frac{1-p}p,q\Bigr) \;<\;p\log(1/p)\,T \+ p(1-p)\,\frac{1-q}2\,.$$
Symmetrizing this with respect to~$p\leftrightarrow1-p$ and substituting into~\eqref{eqGsimple} we get 
the inequality~\eqref{Gineq}.

\section{Discussion}
\label{sec:conclusions}

This paper establishes an upper bound to the gain that informed
traders can extract from their trading activity, in the context of the
Glosten-Milgrom model~\cite{glosten1985bid}. The upper bound is
derived from an analogy with a physical system. This offers the ground
for applying the generalised second law of thermodynamics to financial
systems, suggesting how the no-arbitrage hypothesis can be generalised
in the presence of informed traders.

The key elements in the inequality are the amount of information
$H[Y]$ that informed traders have on the value $Y$ of the asset, the market temperature $T$
and the expected gain $\mathbb{E}[G]$ of informed traders. The market
temperature $T$ measures the level of noise in the market.  It
increases with the fraction of noise traders. It vanishes when these
are absent ($\nu=1$) and it diverges when the fraction of informed
traders vanishes ($\nu\rightarrow 0$). Therefore, (\ref{Gineq}) is very
similar to the bound for work extraction in information engines.
Interestingly, we find that the bound is {attained asymptotically} in the latter
limit, i.e. when $T\to\infty$. This is consistent with the fact that,
in this limit, the convergence of the price to the true value $Y$ and
the activity of informed traders is infinitely slow, as in a
quasi-static process in physics. In the analogy with physics, 
gain extraction approximates a reversible process because at
infinite temperature thermalisation is infinitely fast. Interestingly,
reversibility is the condition that allows maximal work extraction, 
as show in \cite{horowitz2011}. 

A similar analogy has been drawn in Ref.~\cite{vinkler2016analogy}
between the Szil\'ard box and gambling in a sequence of lotteries. In
that case, the optimal work extraction protocol coincides with the
optimal (betting) strategy, and the work extracted with the optimal
rate of growth of the gambler's gain. This suggests that the bound
(\ref{Gineq}) has a wider validity than the framework of
Ref.~\cite{vinkler2016analogy} or of the Glosten-Milgrom model, and it
hints at a generalised second law of thermodynamics for financial
markets. This is a very interesting avenue of further research\footnote{\red{After 
this manuscript first appeared, Pierre Carmier informed us that numerical results
support the conjecture that the inequality Eq.~\eqref{Gineq} extends to finite 
times, if $\mathbb{E}\left[G\right]$ is replaced by the expected gain up to time $t$ and $H[Y]$ 
by the mutual information $I(X_{\le t},Y)$ between $Y$ and the trading activity
up to time $t$.}}.

The present version of the GM model considers a population of an
infinite number of informed traders who behave in a competitive
fashion. The situation is very different from that of a single
informed trader, who trades \red{sequentially. The probability $\nu$ that 
a trade is executed by an informed trader becomes, in this setting, the frequency with which
the informed trader submits orders. The GM model becomes a 
repeated game of incomplete information~\cite{aumann1995repeated} 
between the market maker and the informed trader. The informed trader
can decide the frequency $\nu$ of his orders so as to maximise his gain, 
taking into account that his trading activity reveals information on $Y$ 
to the market maker. By taking an infinitesimally small value of $\nu$, 
the informed trader can access the regime where the inequality~\eqref{Gineq} 
holds asymptotically as an equality, and increase his gain by making the 
market temperature $T(\nu)$ arbitrarily large. Note that the time needed to 
accumulate the gain also diverges in the limit $\nu\to 0$, corresponding to 
a slower impact of the activity of 
the informed trader on the price dynamics. An impatient trader would 
opt for a finite frequency while at the same time leaving a more significant impact
on price's dynamics. Further work in this direction may help shed light on the 
trade-offs between time-constraints and profits, which 
are at the basis of the theories of market impact~\cite{bouchaud2009markets}.}


As the GM model shows, a simple market mechanism can allow private
information to be incorporated into
prices. Ref.~\cite{bardoscia2019lost} argues that financial
transformations, such as diversification or securitisation, degrade
private information in the sense that a considerable fraction of the
side information on financial returns is lost under aggregation. 
Extensions of the present results to multi-asset markets may shed
light on the interplay between diversification and information aggregation.

We hope the present paper will stimulate further research
in the direction of providing an information theoretic basis to
finance.

\section*{Acknowledgments}

L.T. thanks The Abdus Salam International Centre for Theoretical
Physics (ICTP) for hospitality.  We thank Edgar Rold\'an, Gonzalo
Manzano, \red{Jean-Philippe Bouchaud, Pierre Carmier and an 
anonymous referee} for stimulating discussions and comments.

\vskip 1 cm

\bibliographystyle{ieeetr}
\bibliography{ITMMGM.bib}

\end{document}